\documentclass[conference]{IEEEtran}

\usepackage{hyperref}
\usepackage{paralist} 
\usepackage{stmaryrd} 
\usepackage{amsmath}
\usepackage{amsthm}
\usepackage{amssymb}
\usepackage{epsfig} 
\usepackage{epstopdf}
\usepackage{float}
\usepackage{graphicx}
\usepackage{caption, subcaption}

\newtheorem{theorem}{Theorem}
\newtheorem{definition}{Definition}
\newtheorem{lemma}{Lemma}
\newtheorem{proposition}{Proposition}

\newtheorem{example}{Example}

\usepackage{tikz}
\usetikzlibrary{shapes,arrows,fit,calc,positioning,automata}

\newcommand{\eat}[1]{}
\newcommand{\fullversion}[1]{#1}

\makeatletter
\newif\if@restonecol
\makeatother

\usepackage[ruled,vlined,linesnumbered,lined]{algorithm2e}

\usepackage{xcolor,colortbl}
\definecolor{Gray}{gray}{0.95}
\definecolor{LightCyan}{rgb}{0.88,1,1}

 \usepackage{relsize}

\newcommand{\mb}[1]{\ensuremath {\mathbf {#1}}}

\newcommand{\Sk}{\mb{\Psi}}
\newcommand{\sk}{\psi}
\newcommand{\Ska}{{\mb{\Psi^A}}}
\newcommand{\ska}[1]{{\psi_{{#1}}^A}}

\newcommand{\set}[1]{\left\{ #1 \right\}}

\newcommand{\subst}[2]{#1/#2} 

\newcommand{\comment}[1]{\tcp{#1}}

\newcommand{\mcomment}[1]{}

\newcommand{\Sup}{\mathsf{Supp}}

\newcommand{\cb}[2]{\ensuremath{{\mathtt{Cb{#1}}[{#2}]}}}

\newcommand{\cbr}[2]{\ensuremath{{\mathtt{r{#1}}}[{#2}]}}

\newcommand{\true}{\textsf{true}}
\newcommand{\false}{\textsf{false}}
\newcommand{\mono}{\textsc{MonoSkolem}}
\newcommand{\cegar}{\textsc{CegarSkolem}}
\newcommand{\bloqqer}{\textsf{Bloqqer}}

\newcommand{\type}{\textsc{TYPE-}}
\newcommand{\dontprintsemicolon}{} 
\newcommand{\printsemicolon}{} 

\usetikzlibrary{decorations.pathreplacing} 
\usetikzlibrary{patterns}

\usepackage[ruled,vlined,linesnumbered,lined]{algorithm2e}

\tikzset{
block/.style={
  draw, 
  rectangle, 
  minimum height=0cm, 
  minimum width=2cm, align=center,
  }, 
line/.style={->,>=latex'}
}
\tikzstyle{galivenode}=[circle,fill=green!50!black,thick,inner sep=1pt,minimum size=4mm]
\tikzstyle{ralivenode}=[circle,fill=red!70!black,thick,inner sep=1pt,minimum size=4mm]

\tikzstyle{alivenode}=[circle,fill=black!80,thick,inner sep=1pt,minimum size=4mm]
\tikzstyle{deadnode}=[circle,fill=black!10,thick,inner sep=0pt,minimum size=4mm]
\tikzstyle{rdnode}=[circle,draw,fill=red!10,thick,inner sep=2pt,minimum size=4mm]
\tikzstyle{grnode}=[circle,draw,fill=green!10,thick,inner sep=2pt,minimum size=4mm]

\tikzstyle{ynode}=[rectangle,draw=black,fill=yellow!10,thick,inner sep=3pt,minimum size=4mm]
\tikzstyle{gnode}=[rectangle,fill=green!10,thick,inner sep=1pt,minimum size=4mm]
\tikzstyle{bnode}=[rectangle,fill=red!10,thick,inner sep=1pt,minimum size=4mm]

\newcommand{\Red}[1]{#1}

\DeclareMathOperator{\dom}{dom}

\IEEEoverridecommandlockouts

\begin{document}
\title{Skolem Functions for Factored Formulas} 
\author{\IEEEauthorblockN{Ajith K. John}
 \IEEEauthorblockA{HBNI, BARC, India}
  \and
  \IEEEauthorblockN{Shetal Shah}
  \IEEEauthorblockA{IIT Bombay}
  \and
  \IEEEauthorblockN{Supratik Chakraborty}
  \IEEEauthorblockA{IIT Bombay}
  \and
  \IEEEauthorblockN{Ashutosh Trivedi}
  \IEEEauthorblockA{IIT Bombay}
  \and
  \IEEEauthorblockN{S. Akshay}
  \IEEEauthorblockA{IIT Bombay}
}

\maketitle

\begin{abstract}
  Given a propositional formula $F(x, y)$,
  a Skolem function for $x$ is a function $\psi(y)$, such
  that substituting $\psi(y)$ for $x$ in $F$ gives a formula
  semantically equivalent to $\exists x ~F$. Automatically
  generating Skolem functions is of significant interest in several
  applications including certified QBF solving, finding strategies of
  players in games, synthesising circuits and bit-vector programs from
  specifications, disjunctive decomposition of sequential circuits, etc. 
  In many such applications, $F$ is given as a
  conjunction of factors, each of which depends on a small subset of
  variables.  Existing algorithms for Skolem function generation
  ignore any such factored form and treat $F$ as a monolithic
  function.  This presents scalability hurdles in medium to large
  problem instances.  In this paper, we argue that exploiting the
  factored form of $F$ can give significant performance
  improvements in practice when computing Skolem functions.  We
  present a new CEGAR style algorithm for generating Skolem functions
  from factored propositional formulas.  In contrast to earlier work,
  our algorithm neither requires a proof of QBF satisfiability nor
  uses composition of monolithic conjunctions of factors.  We show
  experimentally that our algorithm generates smaller Skolem functions
  and  outperforms 
  state-of-the-art approaches on several large benchmarks. 
\end{abstract}

\section{Introduction}
\label{sec:introduction}

Skolem functions, introduced by Thoralf Skolem in the 1920s, occupy a
central role in mathematical logic.  \Red{Formally, let $F(x,y)$ be a
first-order logic formula, and let $\dom(x)$ and $\dom(y)$ denote the
domains of $x$ and $y$ respectively.} A \emph{Skolem function} for $x$
in $F$ is a function $\psi:\dom(y)\rightarrow \dom(x)$ such that
substituting $\psi(y)$ for $x$ in $F$ yields a formula semantically
equivalent to $\exists x F(x, y)$, i.e. $F(\psi(y), y) \equiv \exists
x F(x, y)$. In this paper, we focus on the case where the formula $F$
is propositional and given as a conjunction of factors.  Classically,
Skolem functions have been used in proving theorems in logic. More
recently, with the advent of fast SAT/SMT solvers, it has been shown
that several practically relevant problems can be encoded as
quantified formulas, and can be solved
\Red{by constructing \emph{realizers} of quantified
variables.  We identify these realizers as specific instances of Skolem
functions, and focus on algorithms for constructing them in this
paper.}

We begin by listing some applications that illustrate the utility
of constructing instances of Skolem functions in practice.
\begin{enumerate}
\item 
\emph{Quantifier elimination}. 
Given a quantified formula $Q x ~F(x, y)$, where
$Q \in \{\exists, \forall\}$, the quantifier elimination problem
requires us to find a quantifier-free formula that is semantically
equivalent to $Q x ~F(x, y)$. Quantifier elimination has important
applications in diverse areas (see, e.g.~\cite{Jian,Gulwani,AMN05} for
a sampling).  It follows from the definition of Skolem function that
eliminating the quantifier from $\exists x F(x, y)$ can be achieved by
substituting $x$ with a Skolem function for $x$.  Since $\forall x
F(x, y)$ can be written as $\neg \exists x \neg F(x, y)$, the same
idea applies in this case too.  In fact, the process can be repeated
in principle to eliminate quantifiers from a formula with arbitrary
quantifier prefix.

\item 
  \emph{Controller Synthesis and Games.}  Control-program synthesis in
  the Ramadge-Wonham~\cite{RW87} framework reduces to games between
  two players---environment and the controller---such that the optimal
  strategy of the controller corresponds to an optimal control
  program.
  The optimal (or winning) strategy of the controller corresponds to 
  choosing 
  values of variables controlled by it such that regardless of the way the
  environment fixes its variables, the resulting play satisfies the controller's
  objective.  
  If the rules of the game are encoded as a propositional formula and if the
  strategy space for both players is finite, the optimal strategy of the
  controller corresponds to finding Skolem functions of variables controlled by it. 
  In fact, for a number of two-player games---such as
  reachability games and safety games~\cite{AMN05},
  tic-tac-toe~\cite{bessiere2006strategic} and chess-like
  games~\cite{ansotegui2005achilles,AMN05}---the problem of deciding a winner can
  be reduced to checking satisfiability of a quantified Boolean formula (QBF), and the problem of finding winning or best-effort
  strategy reduces to Skolem function generation. 

\item 
  \emph{Graph Decomposition.} Skolem functions can be used to compute
  disjunctive decompositions of implicitly specified state transition
  graphs of sequential circuits~\cite{Trivedi}.  The disjunctive
  decomposition problem asks the following question: Given a sequential circuit,
  derive ``component'' sequential circuits, each of which has the same
  state space as the original circuit, but only a subset of
  transitions going out of every state. The components should be such
  that the complete set of state transitions of the original circuit
  is the union of the sets of state transitions of the components.
  Disjunctive decompositions have been shown to be useful in efficient
  reachability analysis~\cite{TCP06}.

\end{enumerate}

There are several other practical applications where Skolem
functions
find use; see, e.g.~\cite{bierre}, for a discussion.  Hence, there is a
growing need for practically efficient and scalable approaches for
generating instances of Skolem functions.
Large and complex representations of the formula $F$ in $\exists x~F$
often present scalability hurdles in generating Skolem functions in
practice.  Interestingly, for several problem instances, the
specification of $F$ is available in a \emph{factored} form, i.e., as
a conjunction of simpler sub-formulas, each of which depends on
a subset of variables appearing in $F$. 
Unfortunately, unlike in the case of disjunction, 
existential quantification does not distribute over conjunction of
sub-formulas. Existing algorithms therefore
ignore any factored form of $F$ and treat the conjunction of factors
as a single monolithic function. We show in this paper that exploiting
the factored form can help significantly when
generating Skolem functions.

Our main technical contribution is a SAT-based Counter-Example Guided
Abstraction-Refinement (CEGAR) algorithm for generating Skolem
functions from factored formulas.  Unlike competing approaches, our
algorithm exploits the factored representation of a formula and
leverages advances made in SAT-solving technology.  The factored
representation is used to arrive at an initial abstraction of Skolem
functions, while a SAT-solver is used as an oracle to identify
counter-examples that are used to refine the Skolem functions
until no counter-examples exist.  
We present a detailed experimental evaluation of our algorithm
vis-a-vis state-of-the-art algorithms~\cite{Jian,bierre} over a
large class of benchmarks.  We show that on several large problem
instances, we outperform competing algorithms. 

\medskip 
\noindent{\bf Related Work.} 
We are not aware of other techniques for Skolem function generation
that exploit the factored form of a formula.  Earlier work on Skolem
function generation broadly fall in one of four categories.  The first
category includes techniques that extract Skolem functions from a
proof of validity of $\exists X ~F(X,
Y)$~\cite{bierre,jiang2,skizzo,squolem}.  In problem instances where
$\exists X~F(X, Y)$ is valid (and this forms an important sub-class of
problems), these techniques can usually find succinct Skolem functions
if there exists a short proof of validity.  However, in several other
important classes of problems, the
formula $\exists X ~F(X, Y)$ does not evaluate to $\true$ for all values of $Y$, and
techniques in the first category cannot be applied.  The second
category includes techniques that use templates for candidate Skolem
functions~\cite{Gulwani}.  These techniques are effective only when
the set of candidate Skolem functions is known and small.  While this
is a reasonable assumption in some domains~\cite{Gulwani}, it is not
in most other domains.  BDD-based techniques~\cite{bdds} are yet
another way to compute Skolem functions.  Unfortunately, these
techniques are known not to scale well, unless custom-crafted variable
orders are used.  The last category includes techniques that use
cofactors to obtain Skolem functions~\cite{Jian,Trivedi}.  These
techniques do not exploit the factored representation of a formula
and, as we show experimentally, do not scale well to large problem
instances.

\section{Preliminaries}
\label{sec:preliminaries}

We use lower case letters (possibly with subscripts) to denote
propositional variables, and upper case letters to denote sequences of
such variables. We use $0$ and $1$ to denote the propositional
constants $\false$ and $\true$, respectively.  Let $F(X, Y)$ be a
propositional formula, where $X$ and $Y$ denote the sequences of
variables $(x_1, \ldots, x_n)$ and $(y_1, \ldots, y_m)$, respectively.
We are interested in problem instances where $F(X, Y)$ is given as a
conjunction  of factors $f^1(X_1, Y_1), \ldots, f^r(X_r, Y_r)$, where
each $X_i$ (resp., $Y_i$) is a possibly empty sub-sequence of $X$
(resp., $Y$).  For notational convenience, we use $F$ and
$\bigwedge_{j=1}^r f^j$ interchangeably throughout this paper.  The
set of variables in $F$ is called the \emph{support} of $F$, and is
denoted $\Sup(F)$. Given a propositional \Red{formula} $F(X)$ and a
propositional function $\Psi(X)$, we use $F[\subst{x_i}{\Psi(X)}]$, or
simply $F[\subst{x_i}{\Psi}]$, to denote the formula obtained by
substituting every occurrence of the variable $x_i$ in $F$ with
$\Psi(X)$. Since the notions of formulas and functions coincide in
propositional logic, the above is also conventionally
called \emph{function composition}.
If $X$ is a sequence of variables and $x_i$ is a variable, we use
$X\setminus x_i$ to denote the sub-sequence of $X$ obtained by
removing $x_i$ (if present) from $X$.  Abusing notation, we use $X$ to
also denote the set of elements in $X$, when there is no confusion.
A \emph{valuation} or \emph{assignment} $\pi$ of $X$ is a mapping
$\pi: X \rightarrow \{0, 1\}$.
\begin{definition}
  Given a propositional \Red{formula} $F(X, Y)$ and a variable $x_i \in X$,
    a \emph{Skolem function} for $x_i$ in $F(X, Y)$ is a function
    $\psi(X\setminus x_i, Y)$ such that $\exists x_i~ F \equiv
    F[\subst{x_i}{\psi}]$.
\end{definition}
A Skolem function for $x_i$ in $F$ need not be unique.  The following
proposition, which effectively follows from~\cite{Jian,Trivedi},
characterizes the space of all Skolem functions for $x_i$ in $F$.
\begin{proposition}
\label{prop:skolem-space}
A function $\psi(X\setminus x_i, Y)$ is a Skolem function for $x_i$ in
$F(X, Y)$ iff $F[\subst{x_i}{1}] \wedge \neg
F[\subst{x_i}{0}] \Rightarrow \psi$ and $\psi \Rightarrow
F[\subst{x_i}{1}] \vee \neg F[\subst{x_i}{0}]$.
\end{proposition}
The function $F[\subst{x_i}{0}]$ (resp., $F[\subst{x_i}{1}]$) is
called the \emph{positive} (resp., \emph{negative}) \emph{cofactor} of
$F$ with respect to $x_i$, and plays a central role in the study of
Skolem functions for propositional formulas.  In particular, it follows from
Proposition~\ref{prop:skolem-space} that $F[\subst{x_i}{1}]$ is a
Skolem function for $x_i$ in $F$.  The above definition for a single variable can be naturally extended to a vector
of variables.  Given $F(X, Y)$, a \emph{Skolem function vector} for
$X=(x_1, \ldots, x_n)$ in $F$ is a vector of functions $\Sk =
(\sk_1, \ldots, \sk_n)$ such that $\exists x_1 \ldots x_n~ F \equiv$
$(\cdots(F[\subst{x_1}{\sk_1}]) \cdots [\subst{x_n}{\sk_n}])$.  A
straightforward way to obtain a Skolem function vector $\Sk$ is to
first obtain a Skolem function $\sk_1$ for $x_1$ in $F$, then compute
$F' \equiv \exists x_1\, F$ and obtain a Skolem function $\sk_2$ for
$x_2$ in $F'$, and so on until $\sk_n$ has been obtained.  More
formally, $\sk_i$ can be computed as a Skolem function for $x_i$ in
$\exists x_1 \ldots x_{i-1}\, F$, starting from $\sk_1$ and proceeding
to $\sk_n$.  Note that $\exists x_1 \ldots x_{i-1}\, F$ can itself be
computed as $\left(\cdots \left(F[\subst{x_1}{\sk_1}]\right) \cdots
[\subst{x_{i-1}}{\sk_{i-1}}]\right)$.
\begin{definition}
\label{def:cb0-cb1}
The ``\mb{C}an't-\mb{b}e-\mb{1}'' function for $x_i$ in $F$, denoted
$\cb{1}{x_i}(F)$, is defined to be $\left(\neg \exists x_1 \ldots x_{i-1}\,
F\right)[\subst{x_i}{1}]$.  Similarly, the
``\mb{C}an't-\mb{b}e-\mb{0}'' function for $x_i$ in $F$, denoted
$\cb{0}{x_i}(F)$, is defined to be $\left(\neg \exists x_1 \ldots x_{i-1}\,
F\right)[\subst{x_i}{0}]$. When $X$ and $F$ are clear from the context,
we  use $\cb{1}{i}$ and $\cb{0}{i}$ for $\cb{1}{x_i}(F)$ and
$\cb{0}{x_i}(F)$, respectively.
\end{definition}
\noindent Intuitively, in order to make $F$ evaluate to $1$, we cannot set $x_i$
to $1$ (resp. $0$) whenever the valuation of $\set{x_{i+1}, \ldots,
x_n} \cup Y$ satisfies $\cb{1}{i}$ (resp., $\cb{0}{i}$).  The
following proposition follows from Definition~\ref{def:cb0-cb1} and
from our observation about computing a Skolem function vector
one component at a time.
\begin{proposition}
  \label{cor:xmaps1} 
$\Sk {=} (\neg\cb{1}{1},  \ldots, \neg\cb{1}{n})$ is a Skolem
  function vector for $X$ in $F$.
\end{proposition}
Note that the support of $\sk_i$ in $\Sk$, as given by
Proposition~\ref{cor:xmaps1}, is $\set{x_{i+1}, \ldots, x_n} \cup Y$.
If we want a Skolem function vector $\Sk$ such that every component
function has only $Y$ (or a subset thereof) as support, this can be
obtained by repeatedly substituting the Skolem function for every
variable $x_i$ in all other Skolem functions where $x_i$ appears.  We 
denote such a Skolem function vector as $\Sk(Y)$.

\section{A monolithic composition based algorithm}
\label{sec:state-of-the-art}
Our algorithm is motivated in part by cofactor-based techniques for
computing Skolem functions, as proposed by Jiang et al~\cite{Jian} and
Trivedi~\cite{Trivedi}.  Given $F(X, Y) = \bigwedge_{j=1}^r f^j(X_j,
Y_j)$, the techniques of~\cite{Jian,Trivedi} essentially compute a
Skolem function vector $\Sk(Y)$ for $X$ in $F$ as shown in algorithm
$\mono$ (see Algorithm~\ref{alg:skolem-generation-mono}). In this
algorithm, the variables in $X$ are assumed to be ordered by their
indices.  While variable ordering is known to affect the difficulty of
computing Skolem functions~\cite{Jian}, \mcomment{more
citations reqd} we assume w.l.o.g.  that the variables are indexed to
represent a desirable order.  We describe the variable order used in
our study later in Section~\ref{sec:expt}.

$\mono$ works in two phases.  In the first phase, it implements a
straightforward strategy for obtaining a Skolem function vector, as
suggested by Proposition~\ref{cor:xmaps1}.
Specifically, steps $3$ and $4$ of $\mono$ build a monolithic
conjunction $F_i$ of all factors that have $x_i$ in their support,
before computing $\sk_i$.  This restricts the scope of the quantifier
for $x_i$ to the conjunction of these factors.
In Step $6$, we use $\neg\cb{1}{i}$ as a specific choice for the
Skolem function $\psi_i$. After computing $\psi_i$ from $F_i$, step
$7$ discards the factors with $x_i$ in their support, and introduces a
single factor representing $\exists x_i\, F_i$ (computed as
$F_i[\subst{x_i}{\psi_i}]$) in their place.  Note that each $\sk_i$
obtained in this manner has $\set{x_{i+1}, \ldots, x_n} \cup Y$ (or a
subset thereof) as support.  Since we want each Skolem function to
have support $Y$, a second phase of ``reverse'' substitutions is
needed.  In this phase (see Algorithm~\ref{alg:reverse}), the Skolem
function $\sk_n(Y)$ obtained above is substituted for $x_n$ in
$\sk_1, \ldots, \sk_{n-1}$.  This effectively renders all Skolem
functions independent of $x_n$.  The process is then repeated with
$\sk_{n-1}$ substituted for $x_{n-1}$ in $\sk_1, \ldots, \sk_{n-2}$
and so on, until all Skolem functions have been made independent of
$x_1, \ldots, x_n$, and have only $Y$ (or subsets thereof) as support.

$\mono$ can be further refined by combining steps 5 and 6, and
directly defining $\psi_i$ in terms of $F_i$. However, we introduce
the intermediate step using $\cb{0}{i}$ and $\cb{1}{i}$ to motivate
their central role in our approach. 
Note that
instead of $\neg\cb{1}{i}$, we could combine $\cb{1}{i}$ and
$\cb{0}{i}$ in other ways (denoted by
\textsc{Combine}($\cb{0}{i}, \cb{1}{i}$) within comments in
Algorithm~\ref{alg:skolem-generation-mono}) to get $\psi_i$ in Step
$6$.  In fact, Jiang et al~\cite{Jian} compute a Skolem function for
$x_i$ in $F$ as an interpolant of $\neg\cb{1}{i} \wedge \cb{0}{i}$ and
$\cb{1}{i} \wedge\neg\cb{0}{i}$, while Trivedi~\cite{Trivedi} observes
that the function $(\neg\cb{1}{i} \wedge(\cb{0}{i} \vee g)) \;\vee\;
(\cb{1}{i}\wedge\cb{0}{i}\wedge h)$ serves as a Skolem function for
$x_i$ in $F$ where $h$ and $g$ are arbitrary propositional functions
with support in $X\setminus\set{x_i} \cup Y$.  Since computing
interpolants using a SAT solver is often time-intensive and does not
always lead to succinct Skolem functions~\cite{Jian}, we simply use
$\neg\cb{1}{i}$ as a Skolem function in Step $6$.
Proposition~\ref{cor:xmaps1} guarantees the correctness of this
choice.

\begin{algorithm}[t]
  \caption{\label{alg:skolem-generation-mono} $\mono$}
  \KwIn{Prop. formula $F(X, Y) = \bigwedge_{j=1}^{r} f^j(X_j, Y_j)$, where $X = (x_1, \ldots, x_n)$}
  \KwOut{Skolem function vector $\Sk(Y)$}
  \tcp{Phase 1 of algorithm}
  $\mathsf{Factors} := \set{f^j \::\: 1 \le j \le r}$\;
  \For{$i$ in $1$ to $n$}{
    $\mathsf{FactorsWithXi} := \set{f \::\: f \in \mathsf{Factors}, x_i \in \Sup(f)}$\;
    $F_i  := \bigwedge_{f \in \mathsf{FactorsWithXi}} f$\; 
    $\cb{0}{i} := \neg F_i[\subst{x_i}{0}]$;  $\cb{1}{i}:=\neg F_i[\subst{x_i}{1}]$\;
    $\sk_i :=$ $\neg\cb{1}{i}$\;
    \tcp{Generally, $\sk_i$:=\textsc{Combine($\cb{0}{i}, \cb{1}{i}$);}}
    $\mathsf{Factors} :=$ 
    $\left(\mathsf{Factors} \setminus \mathsf{FactorsWithXi}\right) \cup
    \set{F_i[\subst{x_i}{\psi_i}]}$\; 
  }
  \tcp{Phase 2 of algorithm}
  \Return{\textsc{ReverseSubstitute}}($\sk_1,\ldots,\sk_n$)\;
\end{algorithm}
\begin{algorithm}[h!]
  \caption{\label{alg:reverse} \textsc{ReverseSubstitute}}
  \KwIn{Functions $\sk_1(x_2,\ldots,x_n,Y), \sk_2(x_3,\ldots,x_n,Y),\ldots, \sk_n(Y)$}
  \KwOut{Function vector $\Sk(Y)$}
  \For{$i=n$ downto $2$}{
    \lFor{$k=i-1$ downto $1$}{$\sk_k$  = $\sk_k[\subst{x_i}{\sk_{i}}]$}}
  \Return $\Sk(Y) = (\sk_1(Y),\ldots, \sk_n(Y))$\;
\end{algorithm} 

Observe that $\mono$ works with a \emph{monolithic} conjunction
($F_i$) of factors that have $x_i$ in their support.  Specifically, it
composes each such monolithic conjunction $F_i$ with a cofactor of
$F_i$ in Step $7$ to eliminate quantifiers sequentially.  This can
lead to large memory footprints 
and more time-outs 
when used with medium to large benchmarks,
as confirmed by our experiments. 
This motivates us to ask if we can
develop a cofactor-based algorithm that does not suffer from the
above drawbacks of $\mono$.

\section{CEGAR for generating Skolem functions}
\label{sec:algo}

We now present a new CEGAR~\cite{CEGAR-JACM} algorithm for generating
Skolem function vectors, that exploits the factored form of $F(X, Y)$.
Like $\mono$, our new algorithm, named $\cegar$, works in two phases,
and assumes that the variables in $X$ are ordered by their indices.
The first phase of the algorithm consists of the core
abstraction-refinement part, and computes a Skolem function vector
$(\sk_1, \ldots, \sk_n)$, where $\sk_i$ has $\set{x_{i+1}, \ldots,
x_n} \cup Y$, or a subset thereof, as support. Unlike in $\mono$, this
phase avoids composing monolithic conjunctions of factors, yielding
simpler Skolem functions.  The second phase of the algorithm
performs reverse substitutions, similar to that in $\mono$.

Before describing the details of $\cegar$, we introduce some additional notation
and terminology.  
Given  propositional functions (or formulas) $f$
and $g$, we say that $f$ \emph{refines} $g$ and  $g$
\emph{abstracts} $f$ iff $f$ logically implies $g$.  
Given $F(X, Y)$ and a vector of functions $\Ska = (\ska{1}, \ldots, \ska{n})$, we
say that $\Ska$ is an \emph{abstract Skolem function vector} for $X$ in $F$ iff
there exists a Skolem function vector $\Sk = (\sk_1, \ldots, \sk_n)$ for $X$ 
in $F$ such that $~\ska{i}$ abstracts $\sk_i$, for every
$i \in \set{1, \ldots, n}$.  Instead of using $\cb{0}{i}$ and $\cb{1}{i}$ to
compute Skolem functions, as was done in $\mono$, we now use their
\emph{refinements}, denoted $\cbr{0}{i}$ and $\cbr{1}{i}$ respectively, to
compute abstract Skolem functions.  For convenience, we represent
$\cbr{0}{i}$ and $\cbr{1}{i}$ as sets of implicitly disjoined
functions.  Thus, if $\cbr{1}{i}$, viewed as a set, is $\set{g_1, g_2}$,
 then it is $g_1 \vee g_2$ when viewed as a function.
We abuse notation and use $\cbr{1}{i}$ (resp., $\cbr{0}{i}$) to
denote a set of functions or their disjunction, as needed.

\subsection{Overview of our CEGAR algorithm} 
\noindent Algorithm $\cegar$ has two phases. The first phase consists of a CEGAR
loop, while the second does reverse substitutions.  The CEGAR
loop has the following steps.
\begin{itemize}
\item {\bf Initial abstraction and refinement.}
  This step involves constructing refinements of $\cb{0}{i}$ and
  $\cb{1}{i}$ for every $x_i$ in $X$.  Using
  Proposition~\ref{cor:xmaps1}, we can then construct an initial
  abstract Skolem function vector $\Ska$. This step is implemented in
  Algorithm~\ref{alg:init-abs-ref} (\textsc{InitAbsRef}), which
  processes individual factors of $F(X, Y)= \bigwedge_{j=1}^r f^j(X_j,
  Y_j)$ separately, without considering their conjunction. As a
  result, this step is time and memory efficient if the individual
  factors are simple with small representations.
\item 
  {\bf Termination Condition.}  Once \textsc{InitAbsRef} has computed
  $\Ska$, we check whether $\Ska$ is already a Skolem function vector.
  This is achieved by constructing an appropriate propositional
  formula $\varepsilon$, called the ``error formula'' for $\Ska$
  (details in Subsection \ref{subsec:termination}), and checking for
  its satisfiability.  An unsatisfiable formula implies that $\Ska$ is
  a Skolem function vector.  Otherwise, a satisfying assignment $\pi$
  of $\varepsilon$ is used to improve the current refinements of
  $\cb{1}{i}$ and $\cb{0}{i}$ for suitable variables $x_i$.
\item 
  {\bf Counterexample guided abstraction and refinement.}  This step
  is implemented in
  Algorithm~\ref{alg:update-abs-ref}: \textsc{UpdateAbsRef}, and
  computes an improved (i.e., more abstract) refinement of
  $\cb{0}{i}$ and $\cb{1}{i}$ for some $x_i \in X$.  This, in turn,
  leads to a refinement of the abstract Skolem function vector $\Ska$.
\end{itemize}
The overall CEGAR loop starts with the first step and repeats the
second and third steps until a Skolem function vector is obtained.
We now discuss the three steps in detail.

\subsection{Initial Abstraction and Refinement} 
\label{subsec:initabsref}

\begin{algorithm}[t]
  \caption{\label{alg:init-abs-ref} \textsc{InitAbsRef}}
  \KwIn{Prop. formula $F(X, Y) = \bigwedge_{j=1}^{r} f^j(X_j, Y_j)$, where
  $X = (x_1, \ldots, x_n)$}
  \KwOut{Abstract Skolem function vector $\Ska = (\ska{1}, \ldots, \ska{n})$, and
    refinements $\cbr{0}{i}$ and $\cbr{1}{i}$ for each $x_i$ in $X$}  
   \For{$i$ in $1$ to $n$}{$\cbr{0}{i}$ := $\emptyset$; $\cbr{1}{i}$ := $\emptyset$; 
                           \tcp{Initializing}}
   \For{$j$ in $1$ to $r$}{
     $f := f^j$; \comment{for each factor }
     \For{$i$ in $1$ to $n$}{
        \If{$x_i \in \Sup(f)$}{
            $\cbr{0}{i} := \cbr{0}{i} \cup \set{\neg f[\subst{x_i}{0}]}$\;
            $\cbr{1}{i} := \cbr{1}{i} \cup \set{\neg f[\subst{x_i}{1}]}$\;
            \comment{Skolem function for $x_i$ in $f$}
            $\psi_{i,f} := f[\subst{x_i}{1}]$\; 
            $f := f[\subst{x_i}{\psi_{i,f}}]$; \tcp{\Red{$\because f[\subst{x_i}{\psi_{i,f}}] \equiv \exists x_i\, f$}}
        }
     }                
   }

   \For{$i$ in $1$ to $n$}{
      $\ska{i} := \neg\cbr{1}{i}$\;
     \comment{Interpreting $\cbr{1}{i}$ as a function}
   }
   \Return {$\Ska{=} (\ska{1}, \ldots, \ska{n})$ and $\cbr{0}{i}, \cbr{1}{i}$ $\forall x_i {\in} X$}
\end{algorithm}
Algorithm \textsc{InitAbsRef} (see Algorithm~\ref{alg:init-abs-ref})
starts by initializing each $\cbr{1}{i}$ and $\cbr{0}{i}$, viewed as
sets, to the empty set.  Subsequently, it considers each factor $f$ in
$\bigwedge_{j=1}^r f^j(X_j, Y_j)$, and determines the contribution of
$f$ to $\cb{0}{i}$ and $\cb{1}{i}$, for every $x_i$ in the support of
$f$.  Specifically, if $x_i \in
\Sup(f)$, the contribution of $f$ to $\cb{0}{i}$ is 
$\left(\neg \exists x_1 \ldots x_{i-1}\,f\right)[\subst{x_i}{0}]$, and
its contribution to $\cb{1}{i}$ is $\left(\neg \exists x_1 \ldots
x_{i-1}\, f\right)[\subst{x_i}{1}]$.  These contributions are
accumulated in the sets $\cbr{0}{i}$ and $\cbr{1}{i}$, respectively,
and $x_i$ is existentially quantified from $f$. The process is then
repeated with the next variable in the support of $f$.  Once the
contributions from all factors are accumulated in $\cbr{0}{i}$ and
$\cbr{1}{i}$ for each $x_i$ in $X$, \textsc{InitAbsRef} computes an
abstract Skolem function $\ska{i}$ for each $x_i$ in $F$ by
complementing $\cbr{1}{i}$, interpreted as a disjunction of functions.
Note that executing steps $4$ through $10$ of \textsc{InitAbsRef} for
a specific factor $f$ is operationally similar to executing steps $1$
through $7$ of $\mono$ with a singleton set of factors, i.e.,
$\mathsf{Factors} = \set{f}$.  This highlights the key difference
between \textsc{InitAbsRef} and $\mono$: while $\mono$ works with
monolithic conjunctions of factors and their
compositions, \textsc{InitAbsRef} works with individual factors,
without ever considering their conjunctions.
Lemma~\ref{lem:initabsref-correct} asserts the correctness of
\textsc{InitAbsRef}.
\begin{lemma}
\label{lem:initabsref-correct}
The vector $\Ska$ computed by \textsc{InitAbsRef} is an abstract
Skolem function vector for $X$ in $F(X, Y)$.  In addition,
$\cbr{0}{i}$ and $\cbr{1}{i}$ computed by \textsc{InitAbsRef} are
refinements of $\cb{0}{i}(F)$ and $\cb{1}{i}(F)$ for every $x_i$ in
$X$.
\end{lemma}
\fullversion{
\begin{proof}
Consider the ordered pair $(j, i)$ of loop indices corresponding to
the nested loops in steps $3-10$ and $5-10$ of
algorithm \textsc{InitAbsRef}.  Every update of $\cbr{0}{i}$ and
$\cbr{1}{i}$ in steps $7$ and $8$ of \textsc{InitAbsRef} can be
associated with a unique ordered pair of loop indices.  Define a
linear ordering $\preceq$ on the loop index pairs as: $(j, i) \preceq
(j', i')$ iff $j < j'$, or $j = j'$ and $i \leq i'$.  Note that this
represents the ordering of loop index pairs in successive iterations
of the loop in steps $5-10$ of \textsc{InitAbsRef}.  We use induction
on $(j, i)$, ordered by $\preceq$, to show that $\cbr{0}{i}$ and
$\cbr{1}{i}$, as computed by \textsc{InitAbsRef}, are refinements of
$\cb{0}{i}$ and $\cb{1}{i}$.  The base case follows from the
initialization in steps $1$ and $2$ of \textsc{InitAbsRef}.  To prove
the inductive step, consider an update of $\cbr{0}{i}$ and
$\cbr{1}{i}$ in steps $7$ and $8$, respectively,
of \textsc{InitAbsRef}.  The function $f$ used in steps $7$ and $8$
is easily seen to be $\exists x_i\ldots x_{i-1}\, f^j$.  Since $f^j$
is a factor of $F$, we also have $F \Rightarrow f^j$.  It follows
that $\exists x_i\ldots x_{i-1}\, F \;\Rightarrow\; \exists x_i\ldots
x_{i-1}\, f^j \;\equiv\; f$.  Taking the contrapositive gives $\neg
f \;\Rightarrow\; \neg\exists x_i\ldots x_{i-1}\, F$.  Therefore,
$\neg f[\subst{x_i}{a}] \Rightarrow (\neg \exists x_1\ldots x_{i-1}\,
F)[\subst{x_i}{a}]$ for every propositional constant $a$.  Recalling
the definitions of $\cb{0}{i}$ and $\cb{1}{i}$, we get $\neg
f[\subst{x_i}{0}] \Rightarrow \cb{0}{i}$ and $\neg
f[\subst{x_i}{1}] \Rightarrow \cb{1}{1}$.  By the inductive
hypothesis, $\cbr{0}{i}$ and $\cbr{1}{i}$ are refinements of
$\cb{0}{i}$ and $\cb{1}{i}$ prior to executing step $7$
of \textsc{InitAbsRef}.  Therefore, the updated values of
$\cbr{0}{i}$ and $\cbr{1}{i}$, as computed in steps $7$ and $8$
of \textsc{InitAbsRef}, are also refinements of $\cb{0}{i}$ and
$\cb{1}{i}$.  This completes the induction.

Since $\cbr{1}{i} \Rightarrow \cb{1}{i}$ for every $x_i$ in $X$ when
we reach step $11$ of \textsc{InitAbsRef}, it follows from
Proposition~\ref{cor:xmaps1} that $\ska{i} = \neg\cbr{1}{i}$ abstracts
a Skolem function for $x_i$ in $F$.  Hence, $\Ska$, as computed by
\textsc{InitAbsRef}, is an abstract Skolem function vector for $X$ in $F$.
\end{proof}
}

\subsection{Termination condition}
\label{subsec:termination}
Given $F(X, Y)$ and an abstract Skolem function vector $\Ska$, it may
happen that $\Ska$ is already a Skolem function vector for $X$ in
$F$.  We therefore check if $\Ska$ is a Skolem function vector
before refinement. 
Towards this end, we define the \emph{error formula} for $\Ska$ as 
$F(X', Y) \wedge \bigwedge_{i=1}^n (x_i \Leftrightarrow \ska{i})
\wedge \neg F(X, Y)$,  where $X' {=} (x_1', \ldots, x_n')$ is a sequence
of fresh variables with no variable in common with $X$.  The
first term in the error formula checks if there exists some valuation
of $X$ that renders $\exists Y F(X,Y)$ true.  The second term assigns
variables in $X$ to the values given by the abstract Skolem functions,
and the third term checks if this assignment falsifies the formula $F$.

\begin{lemma}
\label{lem:term}
The error formula for $\Ska$ is unsatisfiable iff $\Ska$ is a Skolem
function vector of $X$ in $F$.
\end{lemma}
\fullversion{
\begin{proof}
Let $\varepsilon$ be the error formula for $\Ska$.  Suppose
$\varepsilon$ is unsatisfiable. By definition of $\varepsilon$, we
have
\[
\forall Y \forall X' \forall X\,
\left( F(X', Y) \Rightarrow \left(\bigwedge_{i=1}^n (x_i \Leftrightarrow\ska{i}) \;\;\Rightarrow\;\; F(X, Y)\right)\right).
\] 
By standard logic transformations, this implies $\forall Y
\left(\exists X' F(X', Y) \,\Rightarrow\, F'(Y)\right)$, where $F'(Y)$
denotes
$\left(\cdots\left(F[\subst{x_1}{\ska{1}}]\right)\cdots[\subst{x_n}{\ska{n}}]\right)$.
Therefore, $\Ska$ is a Skolem function vector for $X$ in $F$.

Suppose $\pi$ is a satisfying assignment of $\varepsilon$.  By
definition of $\varepsilon$, $\pi$ is a satisfying assignment of
$F(X', Y)$ and of $\bigwedge_{i=1}^n\left(x_i \Leftrightarrow
\ska{i}\right) \wedge \neg F(X, Y)$, considered separately.  Thus, the
values of $x_1, \ldots, x_n$ given by $\ska{1}, \ldots, \ska{n}$
respectively, cause $F$ to evaluate to $0$ for the valuation of $Y$ in
$\pi$.  However, there exists a valuation of $X$ (viz. same as that of
$X'$ in $\pi$) that causes $F$ to evaluate to $1$ for the same
valuation of $Y$ in $\pi$.  Hence, $\Ska$ is not a Skolem function
vector for $X$ in $F$, as witnessed by the valuation of $Y$ in $\pi$.
\end{proof}
}

The following example illustrates the role of the error formula.
\begin{example} \Red{Let $X=\{x_1,x_2\}$, $Y=\{y_1,y_2,y_3\}$ in $\exists x_1 x_2 F(X,Y)$ where $F \equiv 
(f_1\wedge f_2 \wedge f_3)$, with $f_1=(\neg x_1 \vee \neg
x_2 \vee \neg y_1)$, $f_2=(x_2\vee \neg y_3 \vee \neg y_2)$,
$f_3=(x_1\vee \neg x_2\vee y_3)$. }

Algorithm \textsc{InitAbsRef} gives $\cbr{1}{1}=(x_2\wedge y_1)$,
$\cbr{0}{1}=(x_2\wedge \neg y_3)$, $\cbr{1}{2}=\false$,
$\cbr{0}{2}=y_3\wedge y_2$. This yields $\ska{1}=(\neg x_2\vee \neg
y_1)$, $\ska{2}=\true$.  Now, while $\ska{1}$ is a correct Skolem
function for $x_1$ in $F$, $\ska{2}$ is not for $x_2$. This is
detected by the satisfiability of the error formula $\varepsilon =
F(x'_1,x'_2,Y)\wedge (x_1=\neg x_2\vee \neg y_1)\wedge
(x_2=1)\wedge \neg F(x_1,x_2,Y)$.  Note that $\neg F(\neg x_2\vee \neg
y_1,1, Y)$ simplifies to $(y_1 \wedge \neg y_3)$, and
$y_1=1,y_2=1,y_3=0,x_1=0,x_2=1,x_1'=0,x_2'=0$ is a satisfying
assignment for $\varepsilon$.

\end{example}

\subsection{Counterexample-guided abstraction and refinement}
\label{subsec:cegar}

Let $\varepsilon$ be the error formula for $\Ska$, and let $\pi$ be a
satisfying assignment of $\varepsilon$.  We call $\pi$ a
\emph{counterexample} of the claim that $\Ska$ is a Skolem function
vector. For every variable $v \in X' \cup X \cup Y$, we use $\pi(v)$
to denote the value of $v$ in $\pi$.  Satisfiability of $\varepsilon$
implies that we need to refine at least one abstract Skolem function
$\ska{i}$ in $\Ska$ to make it a Skolem function vector.  Since
$\ska{i}$ is $\neg\cbr{1}{i}$ in our approach, refining $\ska{i}$ can
be achieved by computing an improved (i.e., more abstract) version of
$\cbr{1}{i}$.
Algorithm \textsc{UpdateAbsRef} implements this idea by using $\pi$ to
determine which $\cbr{1}{i}$ should be rendered abstract by adding
appropriate functions to $\cbr{1}{i}$, viewed as a set.  

Before delving into the details of \textsc{UpdateAbsRef}, we state
some key results.  In the following, we use $\pi \models f$ to denote
that the formula $f$ evaluates to $1$ when the variables in $\Sup(f)$
are set to values given by $\pi$.  If $\pi \models f$, we also say $f$
evaluates to $1$ under $\pi$.  We use $\cbr{0}{i}_{init}$ and
$\cbr{1}{i}_{init}$ to refer to $\cbr{0}{i}$ and $\cbr{1}{i}$, as
computed by algorithm \textsc{InitAbsRef}.  Since
\textsc{UpdateAbsRef} only adds to $\cbr{1}{i}$ and $\cbr{0}{i}$
viewed as sets, it is easy to see that $\cbr{0}{i}_{init} \Rightarrow
\cbr{0}{i}$ and $\cbr{1}{i}_{init} \Rightarrow \cbr{1}{i}$ viewed as
functions (recall these functions are simply disjunctions of elements
in the corresponding sets).
\begin{lemma}
\label{lem:basis-for-update}
Let $\pi$ be a satisfying assignment of the error formula
$\varepsilon$ for $\Ska$.  Then the following hold.
\begin{enumerate}
\item[(a)] $\pi \models \neg\cb{0}{n} \vee \neg\cb{1}{n}$.
\item[(b)] There exists $k \in \{1, \ldots, n-1\}$ s.t., $\pi \models
  \cbr{1}{k} \wedge \cbr{0}{k}$.
\item[(c)] There exists no Skolem function vector $\Sk = (\sk_1,
  \ldots, \sk_n)$ such that $\sk_j \Leftrightarrow \ska{j}$ for all $j$ in
  $\set{k+1, \ldots, n}$.
\item[(d)] There exists $l \in \{k+1, \ldots, n\}$ such that $x_l = 1$
  in $\pi$, and $\pi \models \cb{1}{l} \wedge \neg\cbr{0}{l}$.
\end{enumerate}
\end{lemma}

\fullversion
{
\begin{proof}
\noindent{\underline{Part (a):}} Consider an assignment $\pi'$ of variables
in $X \cup Y$, such that $\pi'(x_i) = \pi(x_i')$ for all $x_i \in X$,
and $\pi'(y_j) = \pi(y_j)$ for all $y_j \in Y$.  Since $\pi \models
\varepsilon$, by definition of $\varepsilon$, we have $\pi \models
F(X', Y)$. This implies that $\pi' \models F(X, Y)$ and hence, $\pi'
\models \exists x_1 \ldots x_{n-1}\, F$.  If $x_n = 1$ in $\pi'$, we
get $\pi' \models \left(\exists x_1\ldots
x_{n-1}\,F\right)[\subst{x_n}{1}]$, or equivalently, $\pi' \models
\neg\cb{1}{n}$.  If $x_n = 0$ in $\pi'$, by a similar argument, $\pi'
\models \neg\cb{0}{n}$.  Therefore, $\pi' \models \neg\cb{1}{n} \vee
\neg\cb{0}{n}$.  Since $x_n$ is the variable with the highest index in
$X$, both $\cb{1}{n}$ and $\cb{0}{n}$ have only $Y$ as their support.
Since $\pi'(y_j) = \pi(y_j)$ for all $y_j \in Y$, it follows that $\pi
\models \neg\cb{1}{n} \vee \neg\cb{0}{n}$ as well.

\noindent{\underline{Part (b):}} Since $\pi \models \varepsilon$, by
definition of $\varepsilon$, we have $\pi \models \neg F(X, Y)$.  Since
$F = \bigwedge_{q=1}^r f^q$, there exists $j \in \{1, \ldots, r\}$ such
that $\pi \models \neg f^j$.  Without loss of generality, assume that
$\Sup(f^j) \neq \emptyset$ (otherwise, $f^j$ can be removed from
$\bigwedge_{q=1}^r f^q$). Let $x_k$ be the variable with the smallest
index in $\Sup(f^j)$.  We claim that $x_k = 0$ in $\pi$, and prove
this by contradiction.

If possible, let $x_k = 1$ in $\pi$.  Then, $\pi \models (\neg
f^j)[\subst{x_k}{1}]$.  Since $x_k$ is the lowest indexed variable in
$\Sup(f^j)$, it follows from algorithm \textsc{InitAbsRef} that $(\neg
f^j)[\subst{x_k}{1}] \in \cbr{1}{k}_{init}$, when
$\cbr{1}{k}_{init}$ is viewed as a set. This implies that $(\neg
f^j)[\subst{x_k}{1}] \Rightarrow \cbr{1}{k}_{init}$, when
$\cbr{1}{k}_{init}$ is viewed as a function.  Hence, $\pi \models
\cbr{1}{k}_{init}$, and since $\cbr{1}{k}_{init} \Rightarrow
\cbr{1}{k}$, we have $\pi \models \cbr{1}{k}$.  By definition of
$\varepsilon$, we also have $\pi \models (x_k \Leftrightarrow
\ska{k})$, where $\ska{k} = \neg\cbr{1}{k}$.  It follows that $x_k =
\ska{k} = 0$ in $\pi$.  This contradicts our assumption ($x_k = 1$),
and hence $x_k$ must be $0$ in $\pi$.

Since $x_k = 0$ in $\pi$, following the same reasoning as above, we
can show that $\pi \models \cbr{0}{k}$.  Furthermore, since $\pi
\models (x_k \Leftrightarrow \ska{k})$ and $\ska{k} = \neg\cbr{1}{k}$,
having $x_k = 0$ in $\pi$ implies that $\pi \models
\cbr{1}{k}$. Hence, $\pi \models \cbr{0}{k} \wedge \cbr{1}{k}$.  Since
$\cbr{1}{k} \Rightarrow \cb{1}{k}$ and $\cbr{0}{k} \Rightarrow
\cb{0}{k}$, we have $\pi \models \cb{0}{k} \wedge \cb{1}{k}$ as well.
It now follows from part (a) that $k \neq n$ and hence $k \in \{1,
\ldots, n-1\}$

\noindent{\underline{Part (c):}} We prove this by contradiction.  If
possible, let there be a Skolem function vector $\Sk$ such that $\sk_i
\Leftrightarrow \ska{i}$ for all $i$ in $\set{k+1, \ldots, n}$.  Since
$\pi \models F(X', Y)$, it follows that $\pi \models \exists x_1
\ldots, x_n\, F$.  Therefore, by definition of Skolem functions, $\pi
\models \left(\cdots\left(F[\subst{x_1}{\sk_1}]\right) \cdots
        [\subst{x_n}{\sk_n}]\right)$.  Since we have assumed $\sk_i
        \Leftrightarrow \ska{i}$ for all $i$ in $\set{k+1, \ldots, n}$
        and since $\pi \models \bigwedge_{i=1}^n (x_i \Leftrightarrow
        \ska{i})$, it follows that $\pi
        \models\left(\cdots\left(F[\subst{x_1}{\sk_1}]\right)
        \cdots[\subst{x_{k}}{\sk_{k}}]\right)$.  However, we know from
        part (b) that $\pi \models \cbr{0}{k} \wedge \cbr{1}{k}$ and
        hence $\pi \models \cb{0}{k} \wedge \cb{1}{k}$.  Recalling the
        definitions of $\cb{0}{k}$ and $\cb{1}{k}$, we get $\pi
        \models (\neg\exists x_1\ldots x_k\, F)$.  This contradicts
        our inference above, i.e., $\pi
        \models\left(\cdots\left(F[\subst{x_1}{\sk_1}]\right)
        \cdots[\subst{x_{k}}{\sk_{k}}]\right)$.  Hence our assumption
        is wrong, i.e., there is no Skolem function vector $\Sk$ such
        that $\sk_i \Leftrightarrow \ska{i}$ for all $i$ in $\set{k+1,
          \ldots, n}$.

\noindent {\underline{Part (d):}} We prove this by contradiction.  If
possible, suppose $x_l = 0$ in $\pi$, or $\pi \models \neg\cb{1}{l}
\vee \cbr{0}{l}$ for all $l \in \set{k+1, \ldots, n}$.  For convenience
of notation, let us call this
assumption $\mathsf{A}$ in the discussion below.

If $x_l = 0$ in $\pi$, then since $\pi \models \bigwedge_{i=1}^n (x_i
\Leftrightarrow \ska{i})$ and $\ska{i} = \neg\cbr{1}{i}$ for all $i
\in \set{1, \ldots, n}$, it follows that $\pi \models \cbr{1}{l}$.
Since $\cbr{1}{l} \Rightarrow \cb{1}{l}$, we have $\pi \models
\cb{1}{l}$ as well.  It is also easy to see that whenever $\pi \models
\neg\cb{1}{l}$, then $\pi \models \neg\cbr{1}{l}$ as well.  Therefore,
if $x_l = 0$ in $\pi$ or if $\pi \models \neg\cb{1}{l}$, then both
$\cb{1}{l}$ and $\cbr{1}{l}$ evaluate to the same value under $\pi$.

Consider the subcase of assumption $\mathsf{A}$ where $x_l = 0$ in
$\pi$, or $\pi \models \neg\cb{1}{l}$, for all $l \in \set{k+1, \ldots,
  n}$.  From the discussion above, either $\pi \models \cb{1}{l}
\wedge \cbr{1}{l}$ or $\pi \models \neg\cb{1}{l} \wedge
\neg\cbr{1}{l}$ for all $l \in \set{k+1, \ldots, n}$.  Now consider the
Skolem function vector $\Sk$ given by Proposition~\ref{cor:xmaps1}.
Since $\sk_l = \neg\cb{1}{l}$ and $\ska{i} = \neg\cbr{1}{l}$, it
follows that there exists a Skolem function vector, viz. $\Sk$, such
that $\sk_l \Leftrightarrow \ska{l}$ for all $l$ in $\set{k+1, \ldots,
  n}$.  This contradicts the assertion in part (c) above.  Hence we
cannot have $x_l = 0$ in $\pi$ or $\pi \models \neg\cb{1}{l}$, for all
$l \in \set{k+1, \ldots, n}$.

If assumption $\mathsf{A}$ has to hold, there must therefore exist
some $l \in \set{k+1, \ldots, n}$ such that $x_l = 1$ in $\pi$ and $\pi
\models \cb{1}{l}\wedge\cbr{0}{l}$.  Since $\cbr{0}{l} \Rightarrow
\cb{0}{l}$, we must have $\pi \models \cb{1}{l}\wedge\cb{0}{l}$ in
this case.  From part (a), we know that $\pi \models \neg\cb{0}{n}
\wedge \neg\cb{1}{n}$.  It follows that $l$ is strictly less than $n$,
and we can repeat the entire argument above with assumption
$\mathsf{A}$ restricted to indices in $\set{l+1, \ldots, n}$.  Note
that $\set{l+1,\ldots, n}$ is non-empty (since $l < n$), and is a
strict subset of $\set{k+1, \ldots, n}$ (since $l \in \set{k+1, \ldots,
  n}$).  Therefore, restricting assumption $\mathsf{A}$ to smaller
subsets of indices can only be done finitely many times, after which
there won't be any $l$ in the set of indices under consideration such
that $x_l = 1$ in $\pi$ and $\pi \models \cb{1}{l}\wedge\cbr{0}{l}$.
This shows that assumption $\mathsf{A}$ is false, thereby proving the
assertion in part (d).  
\end{proof}
}

\begin{algorithm}[t]
 \caption{\textsc{UpdateAbsRef}}
 \label{alg:update-abs-ref}
  \KwIn{$\cbr{0}{i}$ and $\cbr{1}{i}$ for all $x_i$ in $X$, \\
        \quad \quad \quad Satisfying assignment $\pi$ of error formula, i.e., \\
        \quad \quad \quad \quad $F(X',Y) \wedge \bigwedge_{i=1}^n\left(x_i \Leftrightarrow \ska{i}\right) \wedge \neg F(X, Y)$ } 
  \KwOut{Improved (i.e., refined) $\Ska = (\ska{1}, \ldots, \ska{n})$, \\
         \quad \quad \quad Improved (i.e., abstracted) $\cbr{0}{i}$ $\&$ $\cbr{1}{i}$,  
$\forall x_i\in X$} 
  $k :=$ largest $m$ such that $\pi$ satisfies $\cbr{0}{m} \wedge \cbr{1}{m}$\;
  \dontprintsemicolon
  $\mu_0 :=$ \textsc{Generalize}($\pi$, $\cbr{0}{k}$)\; 
  $\mu_1 :=$ \textsc{Generalize}($\pi$, $\cbr{1}{k}$)\; 
  $\mu := \mu_0 \wedge \mu_1$\; 
  \printsemicolon
  \tcp{Search for Skolem function among $\set{\ska{k+1}, \ldots, \ska{n}}$ to be refined} 
  $l :=$ $k+1$\;
  \While(\tcp*[f]{current guess: \Red{refine $\ska{l}$}}) {\true} 
      {
          \If{$x_l \in \Sup(\mu)$}
             {
               \If {$x_l = 1$ in $\pi$}
                  { 
                    \dontprintsemicolon
                    $\mu_1 := \mu[\subst{x_l}{1}]$\; 
                    $\cbr{1}{l} :=$ $\cbr{1}{l}$ $\cup$ $\set{\mu_1}$\; 
                    \printsemicolon
                    \If {$\pi$ satisfies  $\cbr{0}{l}$}
                       { 
                         \dontprintsemicolon
                         $\mu_0 := $ \textsc{Generalize}($\pi$, $\cbr{0}{l}$)\;
                         $\mu :=$ $\mu_0 \wedge \mu_1$\; 
                         \printsemicolon
                       }
                     \Else {
                         \dontprintsemicolon
                         {\bfseries break}\; 
                         \printsemicolon
                      }
                  }
                  \Else {
                    \dontprintsemicolon
                    $\mu_0:= \mu[\subst{x_l}{0}]$\; 
                    $\cbr{0}{l} :=$ $\cbr{0}{l}$ $\cup$ $\set{\mu_0}$\; 
                    $\mu_1 := $ \textsc{Generalize}($\pi$, $\cbr{1}{l}$)\;
                    $\mu :=$ $\mu_0 \wedge \mu_1$\; 
                   \printsemicolon
                  }
             }
             $l$ $:=$ $l+1$ \;
        }
        $\Ska = (\neg\cbr{1}{1}, \ldots, \neg\cbr{1}{n})$\;
        \Return {$\cbr{0}{i}$ and $\cbr{1}{i}$ for all $x_i$ in $X$, and $\Ska$}
\end{algorithm}

Algorithm~\ref{alg:update-abs-ref} (\textsc{UpdateAbsRef}) uses
Lemma~\ref{lem:basis-for-update} to compute abstract versions of
$\cbr{0}{i}$ and $\cbr{1}{i}$, and a refined version of $\Ska$, when
$\Ska$ is not a Skolem function vector. It takes as input the current versions of $\cbr{0}{i}$ and $\cbr{1}{i}$ for all $x_i$ in
$X$, and a satisfying assignment $\pi$ of the error formula for the
current version of $\Ska$.  Since $\pi \models F(X',Y)$ and $\pi
\models \neg F(X, Y)$, and since the value of every $x_i$ in $\pi$ is
given by $\ska{i}$, there exists at least one $\ska{l}$, for $l \in
\set{1, \ldots, n}$, that fails to generate the right value of $x_l$
when the value of $Y$ is as given by $\pi$.  \textsc{UpdateAbsRef}
works by identifying such an index $l$ and refining $\ska{l}$.  Since
$\ska{i} = \neg \cbr{1}{i}$, $\ska{l}$ is refined by updating (abstracting) the corresponding $\cbr{1}{l}$ set.  In
fact, the algorithm may, in general, end up abstracting not only
$\cbr{1}{l}$, but several $\cbr{0}{i}$ and $\cbr{1}{i}$ as well in
a sound manner.

As shown in Algorithm~\ref{alg:update-abs-ref}, \textsc{UpdateAbsRef}
first finds the largest index $k$ such that $\pi \models
\cbr{0}{k}\wedge\cbr{1}{k}$.  Lemma~\ref{lem:basis-for-update}b
guarantees the existence of such an index in $\set{1, \ldots, n}$.  We
assume access to a function called \textsc{Generalize} that takes as
arguments an assignment $\pi$ and a function $\varphi$ such that $\pi
\models \varphi$, and returns a function $\xi$ that generalizes $\pi$
while satisfying $\varphi$.  More formally, if $\xi =$
\textsc{Generalize}($\pi$, $\varphi$), then $\Sup(\xi) \subseteq
\Sup(\varphi)$, $\pi \models \xi$ and $\xi \Rightarrow \varphi$
(details of \textsc{Generalize} used in our implementation are
discussed later).  Thus, in steps $2$ and $3$ of
\textsc{UpdateAbsRef}, we compute generalizations of $\pi$ that
satisfy $\cbr{0}{k}$ and $\cbr{1}{k}$, respectively.  The function
$\mu$ computed in step $4$ is therefore such that $\pi \models \mu$
and $\mu \Rightarrow \cbr{0}{k} \wedge \cbr{1}{k}$.  Since $\cbr{0}{k}
\wedge \cbr{1}{k} \Rightarrow \neg\exists x_1\ldots x_{k} F$, any
abstract Skolem function vector that produces values of $x_1, \ldots,
x_n$ (given the valuation of $Y$ as in $\pi$) for which $\mu$ evaluates to
$1$, cannot be a Skolem function vector.  Since the support of $\mu$
is $\set{x_{k+1}, \ldots, x_n} \cup Y$, one of the abstract Skolem
functions $\ska{k+1}, \ldots, \ska{n}$ must be refined.

The loop in steps $6$--$21$ of \textsc{UpdateAbsRef} tries to identify an
abstract Skolem function $\ska{l}$ to be refined, by iterating $l$
from $k+1$ to $n$.  Clearly, if $x_l \not\in \Sup(\mu)$, the value of
$\ska{l}$ under $\pi$ is of no consequence in evaluating $\mu$, and we
ignore such variables.  If $x_l \in \Sup(\mu)$ and if $x_l = 1$ in
$\pi$, then $\pi \models \mu[\subst{x_l}{1}]$ and $\mu[\subst{x_l}{1}]
\Rightarrow (\neg\exists x_1\ldots x_{l-1} F)[\subst{x_l}{1}]$.
Recalling the definition of $\cb{1}{l}$, we have $\mu[\subst{x_l}{1}]
\Rightarrow \cb{1}{l}$, and therefore $\mu[\subst{x_l}{1}]$ can be
added to $\cbr{1}{l}$ (viewed as a set) yielding a more abstract
version of $\cbr{1}{l}$.  Steps $8$--$10$ of \textsc{UpdateAbsRef}
implement this update of $\cbr{1}{l}$. 
Note that since $\pi \models \mu[\subst{x_l}{1}]$, we have $\pi
\models \cbr{1}{l}$ after step $10$.  If it so happens that $\pi
\models \cbr{0}{l}$ as well, then we have $\pi \models \cbr{0}{l}
\wedge \cbr{1}{l}$, where $\cbr{1}{l}$ refers to the updated
refinement of $\cb{1}{l}$.  In this case, we have effectively found an
index $l > k$ such that $\pi \models \cbr{0}{k} \wedge \cbr{1}{k}$.
We can therefore repeat our algorithm starting with $l$ instead of
$k$.  Steps $11$--$13$ followed by step $21$ of algorithm
\textsc{UpdateAbsRef} effectively implement this.  If, on the other
hand, $\pi \not\models \cbr{0}{k}$, then we have found an $l$ that
satisfies the conditions in Lemma~\ref{lem:basis-for-update}d.  We
exit the search for an abstract Skolem function in this case (see
steps $14$--$15$).

If $x_l = 0$ in $\pi$, a similar argument as above shows that
$\mu[\subst{x_l}{0}]$ can be added to $\cbr{0}{l}$.  Steps $17$--$18$ of
\textsc{UpdateAbsRef} implement this update.  As before, it is easy to
see that $\pi \models \cbr{0}{l}$ after step $18$.  Moreover, since
$\pi \models \bigwedge_{i=1}^n(x_i \Leftrightarrow \ska{i})$ and
$\ska{i} \equiv \neg\cbr{1}{l}$, in order to have $x_l = 0$ in $\pi$,
we must have $\pi \models \cbr{1}{l}$.  Therefore, we have once again
found an index $l > k$ such that $\pi \models \cbr{0}{k} \wedge
\cbr{1}{k}$, and can repeat our algorithm starting with $l$ instead of
$k$.  Steps $19$--$21$ of algorithm \textsc{UpdateAbsRef} effectively
implement this. 

Once we exit the loop in steps $6$--$21$ of \textsc{UpdateAbsRef}, we
compute the refined Skolem function vector $\Ska$ as $(\neg\cbr{1}{1},
\ldots \neg\cbr{1}{n})$ in step $22$ and return the updated
$\cbr{0}{i}$, $\cbr{1}{i}$ for all $x_i$ in $X$, and also $\Ska$.

\fullversion{\begin{lemma}
\label{lem:update-absref-correct}
Algorithm \textsc{UpdateAbsRef} always terminates, and renders 
at least one $\cbr{1}{i}$ strictly abstract, and at least
one $\ska{i}$ strictly refined, for $i \in \set{1, \ldots, n}$.
\end{lemma}
\begin{proof}
By Lemma~\ref{lem:basis-for-update}a, we know that $\pi \models
\neg\cb{0}{n} \vee \neg\cb{1}{n}$, and therefore $\pi \models
\neg\cbr{0}{n} \vee \neg\cbr{1}{n}$. Since steps $12$--$13$ or
$17$-$20$ of \textsc{UpdateAbsRef} can be executed only when $\pi
\models \cbr{0}{l} \wedge \cbr{1}{l}$, and since $l$ is incremented in
every iteration of the loop in steps $6$--$21$, it follows that steps
$14$--$15$ must be executed for some $l \le n$.  Therefore, algorithm
\textsc{UpdateAbsRef} always terminates.

It is easy to see from the pseudocode of algorithm
\textsc{UpdateAbsRef} that steps $7$--$10$ and $14$-$15$ must be
executed before exiting the while loop (steps $6$--$21$) and
terminating.  Before executing step $10$, we have $x_l = 1$ in $\pi$
and $\pi \models \bigwedge_{i=1}^n(x_i \Leftrightarrow \ska{i})$ .
Since $\ska{l} \equiv \neg\cbr{1}{l}$ before step $10$, with $x_l = 1$
in $\pi$, it must be the case that $\pi \models \neg\cbr{1}{l}$ before
step $10$.  However, since $\pi \models \mu[\subst{x_l}{1}]$ in step
$9$, we have $\pi {\models} \cbr{1}{l}$ after step $10$.  Therefore,
executing step $10$ renders $\cbr{1}{l}$ strictly abstract than
what it was earlier.  This also implies that $\ska{l} \equiv
\neg\cbr{1}{l}$ is strictly refined when \textsc{UpdateAbsRef}
returns in step~$23$.  
\end{proof}
}

\setcounter{example}{0}
\begin{example}[Continued]
\Red{Continuing with our earlier example, the error formula after the
  first step has a satisfying assignment
  $y_1=1,y_2=1,y_3=0,x_1=0,x_2=1,x_1'=0,x_2'=0$. Using this for $\pi$
  in \textsc{UpdateAbsRef}, we find that $\ska{1}$ is left unchanged
  at $(\neg x_2 \vee \neg y_1)$, while $\ska{2}$, which was $\true$
  earlier, is refined to $(\neg y_1\vee y_3)$. With these refined
  Skolem functions, $F(\ska{1}, \ska{2},Y)$ evaluates to $\true$ for
  all valuations of $Y$. As a result, the (new) error formula becomes
  unsatisfiable, confirming the correctness of the Skolem functions.}
\end{example}
 The overall
$\cegar$ algorithm can now be implemented as depicted in
Algorithm~\ref{alg:overall-cegar}. From the above discussion and
Lemmas~\ref{lem:initabsref-correct}, \ref{lem:term} and \ref{lem:update-absref-correct}, we obtain our
main result.

\begin{algorithm}[t]
 \caption{$\cegar$}
 \label{alg:overall-cegar}
  \KwIn{Propositional formula $F(X, Y) = \bigwedge_{j=1}^r f^j(X_j, Y_j)$, $X = (x_1, \ldots, x_n)$} 
  \KwOut{Skolem function vector $\Sk(Y)$ for $X$ in $F$} 
  $(\Ska$,$\set{\cbr{0}{i}, \cbr{1}{i}: 1 \le i \le n}) := \textsc{InitAbsRef}(\bigwedge_{j=1}^r f^j)$\;
  $\varepsilon :=$ $F(X', Y) \wedge \bigwedge_{i=1}^n (x_i \Leftrightarrow \ska{i}) \wedge \neg F(X, Y)$\;
  \While {$\varepsilon$ is satisfiable} {
     Let $\pi$ be a satisfying assignment of $\varepsilon$\;
     ($\Ska$, $\set{\cbr{0}{i}, \cbr{1}{i} \::\: 1 \le i \le n}$) $:=$ \textsc{UpdateAbsRef}($\set{\cbr{0}{i}, \cbr{1}{i} \::\: 1 \le i \le n}$, $\pi$)\;
  $\varepsilon :=$ $F(X', Y) \wedge \bigwedge_{i=1}^n (x_i \Leftrightarrow \ska{i}) \wedge \neg F(X, Y)$\;
 }
 $ \Sk(Y) := \textsc{ReverseSubstitute}(\neg\cbr{1}{1}, \ldots, \neg\cbr{1}{n})$\;
 \Return {$\Sk(Y)$}\;
\end{algorithm}

\begin{theorem}
\label{thm:cegar-correct}
$\cegar$($F(X, Y)$) terminates and computes a Skolem function vector for $X$ in $F$.
\end{theorem}

\fullversion{
\begin{proof}
By Lemma~\ref{lem:update-absref-correct}, we know that every
invocation of \textsc{UpdateAbsRef} renders at least one $\cbr{1}{i}$
strictly abstract than what it was earlier.  Since $\cbr{1}{i}$ is a
propositional function, it has finitely many minterms and can be
rendered strictly abstract only finitely many times.  From
Proposition~\ref{cor:xmaps1}, we also know that $(\neg\cb{1}{1},
\ldots, \neg\cb{1}{n})$ is indeed a Skolem function vector, and
therefore by Lemma~\ref{lem:term}, its error formula is
unsatisfiable. The termination of $\cegar$ follows immediately from
the above observations.  Since $\varepsilon$ is unsatisfiable when
$\cegar$ terminates, it follows from Lemma~\ref{lem:term} that the
vector of functions returned is a Skolem function vector for $X$ in
$F$.
\end{proof}
}

The function \textsc{Generalize}($\pi$, $\varphi$) used in
\textsc{UpdateAbsRef} can be implemented in several ways.  Since $\pi
\models \varphi$, we may return a conjunction of literals
corresponding to the assignment $\pi$, or the function $\varphi$
itself.  From our experiments, it appears that the first option leads
to low memory requirements and increased run-time (due to large number
of invocations of \textsc{UpdateAbsRef}). The other option requires
more memory and less run-time due to fewer invocations of
\textsc{UpdateAbsRef}.  For our study, we let
\textsc{Generalize}($\pi$, $\cbr{1}{k}$) \Red{return one element
in $\cbr{1}{k}$ (viewed as a set) amongst all those that evaluate to
$1$ under $\pi$, such that the support of $\mu$ computed in Algorithm
\textsc{UpdateAbsRef} is minimized (we had to allow
\textsc{Generalize}$(\cdot,\cdot)$ access to $\mu$ for this purpose)}.
We follow a similar strategy for \textsc{Generalize}($\pi$,
$\cbr{0}{k}$).  This gives us a reasonable tradeoff between time and
space requirements.

\section{Experimental Results} \label{results}
\label{sec:expt}

\subsection{Experimental Methodology}
We compared $\cegar$ with (a) $\mono$ (the algorithm based on the
cofactoring approach of~\cite{Jian,Trivedi}) and with (b) $\bloqqer$
(a QRAT-based Skolem function generation tool reported
in~\cite{bierre}). As described in~\cite{bierre}, $\bloqqer$ generates
Skolem functions by first generating QRAT proofs using a remarkably
efficient (albeit incomplete) preprocessor, and then generates Skolem
functions from these proofs.

The Skolem function generation benchmarks were obtained by considering
sequential circuits from the HWMCC10 benchmark suite, and by reducing
the problem of disjunctively decomposing a circuit into components to
the problem of generating Skolem function vectors.  Details of how
these benchmarks were generated are described
in~\cite{benchmarks}.  Each benchmark is of the form $\exists X
F(X,Y)$, where $F(X,Y)$ is a conjunction of factors and $\exists Y
(\exists X F(X,Y))$ is $\true$. However, for some benchmarks, $\forall
Y (\exists X F(X,Y))$ does not evaluate to $\true$.  Since 
$\bloqqer$ can generate Skolem functions only when
$\forall Y (\exists X F(X,Y))$ is $\true$, 
 we divided the benchmarks into two categories: a) {\type1} where
 $\forall Y \exists X F(X,Y)$ is $\true$, and b) {\type2} where
 $\forall Y \exists X F(X,Y)$ is $\false$ (although $\exists Y \exists
 X F(X,Y)$ is $\true$). While we ran $\cegar$ and $\mono$ on all
 benchmarks, we ran $\bloqqer$ only on $\type1$ benchmarks. Further,
 since $\bloqqer$ required the input to be in \texttt{qdimacs} format,
 we converted each $\type1$ benchmark into \texttt{qdimacs} format
 using Tseitin encoding~\cite{tseitin68}. All our benchmarks can be
 downloaded from~\cite{benchmarks}.

Our implementations of $\mono$ and $\cegar$ make use of the
ABC~\cite{abc-tool} library to represent and manipulate functions as
AIGs.  For $\cegar$, we used the default SAT solver provided by ABC,
which is a variant of MiniSAT.
We used a simple heuristic to order the variables, 
 and used the same ordering for both $\mono$ and $\cegar$.
In our ordering, variables that occur in fewer factors are indexed lower than
those that occur in more factors. 

We used the following metrics to compare the performance of the
algorithms: (i) average/maximum size of the generated Skolem functions
in a Skolem function vector, where the size is the number of nodes in
the AIG representation of a function, and ii) total time taken to
generate the Skolem function vector (excluding any input format
conversion time).  The experiments were performed on a 1.87 GHz
Intel(R) Xeon machine with $128$GB memory running Ubuntu 12.04.4.  The
maximum time and main memory usage was restricted to 2 hours and
32GB, although we noticed that for most benchmarks, all three
algorithms used less than $2$ GB memory.

\subsection{Results and Discussion}
We conducted our experiments with $424$ benchmarks, of which $160$ were $\type1$ benchmarks and $264$ were $\type2$ benchmarks. 
The $424$ benchmarks covered a wide spectrum in terms of number of factors, total number of variables, and number of quantified variables. 
\fullversion{
For instance, in  the $\type1$ category, the number of factors varied
from $44$ to $7034$, total number of variables varied from $94$ to $9782$ and 
 the number of variables to eliminate varied from $60$ to $4751$.
  Amongst the $\type2$ benchmarks, the number of factors varied 
 varied from $24$ to $3956$,  the total number of varibles varied from $70$ to
 $5963$, and the variables to eliminate varied $21$ to $2689$.  
}
\subsubsection{$\cegar$ vs $\mono$}
The performance of these two algorithms on all the 
 benchmarks ($\type1$ and $\type2$) is shown in the scatter plots of
Figure~\ref{fig:cegar1}, where  Figure~\ref{fig:cegar_size} shows
the average sizes of Skolem functions generated in a Skolem function
vector and Figure~\ref{fig:cegar_time} shows the total time taken in seconds.
From Figure~\ref{fig:cegar_size}, it is clear that the Skolem
functions generated by $\cegar$ in a Skolem function vector are on
average {\em smaller} than those generated by $\mono$. {\it There is
  no instance on which $\cegar$ generates Skolem function vectors with
  larger functions on average vis-a-vis $\mono$}.
\begin{figure}[t]
\centering
\begin{subfigure}{.49\textwidth}
  \centering
  \includegraphics[angle=-90,  width=.9\linewidth, scale=0.27]{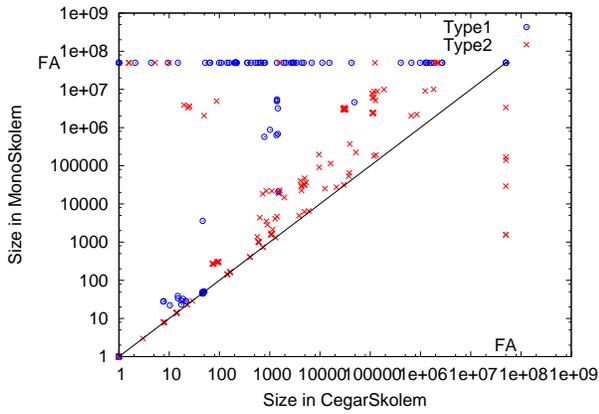}
  \caption{Average Skolem function sizes}
  \label{fig:cegar_size}
\end{subfigure}
\begin{subfigure}{.49\textwidth}
  \centering
  \includegraphics[angle=-90, width=.9\linewidth, scale=0.27]{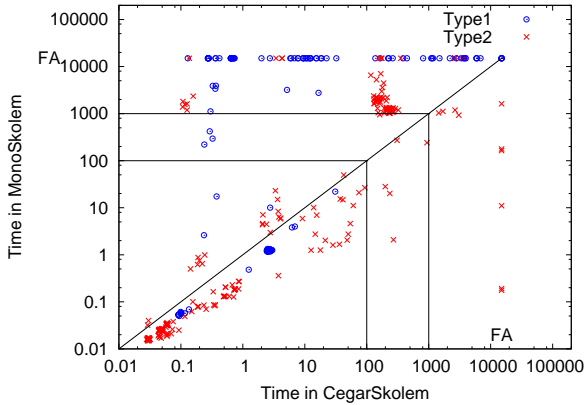}
  \caption{Time taken (in seconds) }
  \label{fig:cegar_time}
\end{subfigure}%
\caption{\small{$\cegar$ vs $\mono$ on $\type1$ \& $\type2$ benchmarks. Topmost (rightmost) points indicate benchmarks where $\mono$ ($\cegar$) was unsuccessful.} } 
\vspace*{-1em}
\label{fig:cegar1}
\end{figure}

Due to repeated calls to the SAT-solver, $\cegar$ takes more time than $\mono$ on some benchmarks, but on most of them the total time taken by both algorithms is {\it less than} $100$ seconds (Figure \ref{fig:cegar_time}). Indeed, on profiling we found that $\cegar$ spent most of its time on SAT solving. On $38$ benchmarks where $\cegar$ took greater than $100$ but less than $300$ seconds,
$\mono$ performed significantly worse, taking more than $1000$ seconds. We found the degradation of $\mono$ was due to the large sizes of Skolem functions generated (of the order of $1$ million AIG nodes) compared to those generated by $\cegar$  ($< 8000$ AIG nodes). \emph{Large Skolem function sizes clearly imply more time spent in function composition and reverse-substitution.}

For benchmarks where the sizes of Skolem
  functions generated were even larger (of the order of $10^7$ AIG nodes),
  $\mono$ could not complete generation of all Skolem functions:  for $8$ benchmarks, the memory consumed by $\mono$ increased rapidly, resulting
  in  memory outs;  for $10$ benchmarks, it ran out of time; for an
  overwhelming $83$ benchmarks, it encountered integer overflows (and
  hence assertion failures) in the underlying ABC library. These are
  indicated by the topmost points (see label ``FA'' on the axes)  in Figure \ref{fig:cegar1}. 
In contrast,  {\em  $\cegar$ generated Skolem functions for almost all} $(412/424)$ {\em benchmarks}.
  The rightmost points indicate the $12$ cases where $\cegar$ failed, of which $10$ were time-outs and $2$ were memory outs. 

\subsubsection{$\cegar$ vs $\bloqqer$}
\begin{figure}[t]
  \centering
\begin{subfigure}{.49\textwidth}
  \centering
  \includegraphics[angle=-90,  width=.9\linewidth, scale=0.25]{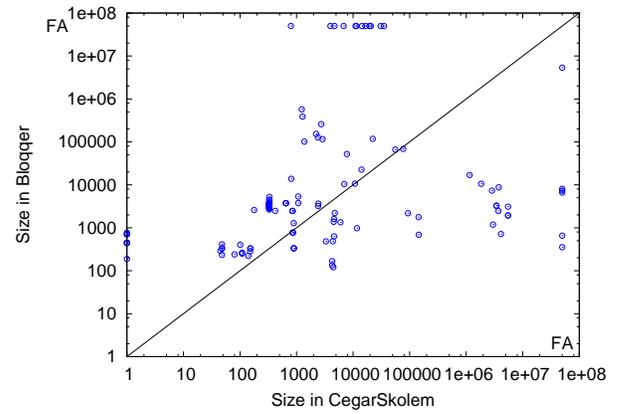}
  \caption{Maximum size of Skolem functions}
  \label{fig:cegar_bloqqer_size}
\end{subfigure}%

\begin{subfigure}{.49\textwidth}
  \includegraphics[angle=-90,  width=.9\linewidth, scale=0.27]{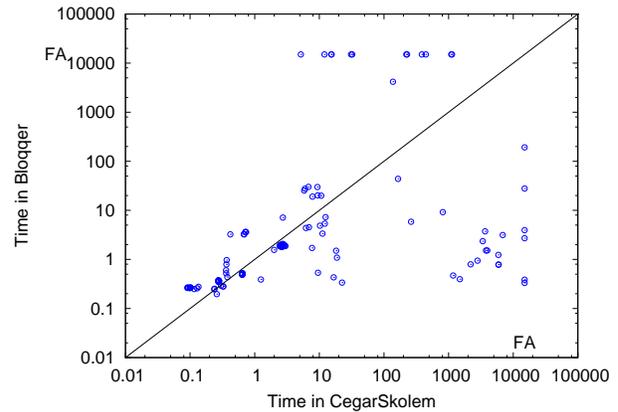}
  \caption{Time taken (in seconds)}
  \label{fig:cegar_bloqqer_time}
\end{subfigure}

\caption{\small{$\cegar$ vs $\bloqqer$ on $\type1$ benchmarks. Topmost (rightmost) points indicate benchmarks for which $\bloqqer$ ($\cegar$)  was unsuccessful.}} 
\vspace*{-1em}
\label{fig:cegarbloqqer1}
\end{figure}

Of the $160$ $\type1$ benchmarks, $\bloqqer$ successfully generated
Skolem function vectors in $148$ cases. It gave a {\texttt {NOT
    VERIFIED}} message for the remaining $12$ benchmarks (in less than
30 minutes). These benchmarks are indicated by the topmost points (see
label ``FA'' on the axes) in the scatter plots of Figure
\ref{fig:cegarbloqqer1}.  Of these, $8$ are large benchmarks with
$1000+$ factors and variables to eliminate (overall, there are $9$
such large benchmarks). On the other hand, $\cegar$ was able to
successfully generate Skolem functions on $154$ benchmarks, including
the $9$ large benchmarks, on each of which it took less than $20$
minutes.

For the $142$ benchmarks for which both algorithms succeeded, we compared the 
times taken in Figure~\ref{fig:cegar_bloqqer_time}. As earlier, $\cegar$ took more time on many benchmarks, but there were several benchmarks, including the large benchmarks, on which $\bloqqer$ was out-performed. 
We also compared the maximum sizes of Skolem functions generated in a Skolem function vector (see Figure~\ref{fig:cegar_bloqqer_size}). We used the maximum (instead of average) size, since Tseitin encoding was needed to convert the benchmarks to \texttt{qdimacs} format, and this introduces many variables whose Skolem function sizes are very small, skewing the average.
For a majority  ($108/142$) of the benchmarks where both algorithms succeeded, the maximum sizes of Skolem functions obtained by $\cegar$ were {\em  smaller} than those generated by $\bloqqer$. 
Hence, \emph{not only does $\cegar$ run faster on the large benchmarks, it also generates smaller Skolem functions on most of them}.

\subsubsection{Discussion}
  For all benchmarks on which $\cegar$  timed out, 
  we noticed that 
  there were large subsets of factors that shared many variables in
  their supports.  As a result, $\cegar$ could not exploit the
  factored representation effectively, requiring many refinements.
   We also noticed that for many benchmarks ($197/424$), the initial abstract Skolem functions were correct, and most of the time was spent in the SAT solver.
In fact, on averaging over all benchmarks, we found that around
$33\%$ of the time spent by $\cegar$ was for SAT-solving.  
This shows that we can leverage improvements in SAT solving technology
to improve the performance of $\cegar$.  

\section{Conclusion and Future Work}
\label{sec:conclusion}

We presented a CEGAR algorithm for generating Skolem functions from
factored propositional formulas.  Our experiments show that for
complex functions, our algorithm out-performs two
state-of-the-art algorithms. As part of future work, we will explore integration with more efficient SAT-solvers and refinement using multiple counter-examples.
\bibliographystyle{plain} 
\bibliography{papers}

\end{document}